\documentclass[letterpaper,10 pt,conference]{ieeeconf}  

\IEEEoverridecommandlockouts                              

\renewcommand{\baselinestretch}{0.94}

\makeatletter
\def\endthebibliography{%
	\def\@noitemerr{\@latex@warning{Empty `thebibliography' environment}}%
	\endlist
}
\makeatother

\usepackage{cite}
\usepackage{amsmath,amssymb,amsfonts}
\usepackage{algorithmic}
\usepackage{graphicx}
\usepackage{textcomp}
\usepackage{xcolor}

\usepackage{soul}

\newtheorem{remark}{\textbf{Remark}}
\newtheorem{lem}{\textbf{Lemma}}
\newtheorem{theorem}{\textbf{Theorem}}

\newtheorem{assumption}{\textbf{Assumption}}
\newtheorem{definition}{\textbf{Definition}}
\newtheorem{corollary}{\textbf{Corollary}}

\title{\LARGE \bf
Undetectable Cyber Attacks on Communication Links in Multi-Agent Cyber-Physical Systems}

\author{Mahdi Taheri$^{1}$, Khashayar Khorasani$^{1}$, Iman Shames$^{2}$, and Nader Meskin$^{3}$
\thanks{$^{1}$Mahdi Taheri (m$\_$eri@encs.concordia.ca) and Khashayar Khorasani (kash@ece.concordia.ca) are with the Department of Electrical and Computer Engineering, Concordia University, Montreal, Canada.}%
\thanks{$^{2}$Iman Shames (iman.shames@unimelb.edu.au) is with the Department of Electrical and Electronic Engineering, University of Melbourne, Melbourne, Australia.}%
\thanks{$^{3}$Nader Meskin (nader.meskin@qu.edu.qa) is with the Department of Electrical Engineering, Qatar University, Doha, Qatar.}%
\thanks{The authors would like to acknowledge the financial support received from NATO under the Emerging Security Challenges Division program. K. Khorasani and N. Meskin would like to acknowledge the support received from NPRP grant number 10-0105-17017 from the Qatar National Research Fund (a member of Qatar Foundation). K. Khorasani would also like to acknowledge the support received from the Natural Sciences and Engineering Research Council of Canada (NSERC) and the Department of National Defence (DnD) under the Discovery Grant and DnD Supplemental Programs. The statements made herein are solely the responsibility of the authors.}
}

\begin{document}

\maketitle
\thispagestyle{empty}
\pagestyle{empty}

\begin{abstract}
	The objective in this paper is to study and develop conditions for a network of multi-agent cyber-physical systems (MAS) where a malicious adversary can utilize  vulnerabilities in order to ensure and maintain cyber attacks undetectable. We classify these  cyber attacks as undetectable in the sense that their impact cannot be observed in the generated residuals. It is shown if an agent that is the root of a rooted spanning tree in the MAS graph is under a cyber attack, the attack is undetectable by the entire network. Next we investigate if a non-root agent is compromised, then under certain conditions cyber attacks can become detectable. Moreover, a novel cyber attack that is designated as quasi-covert cyber attack is introduced that can be used to eliminate detectable impacts of cyber attacks to the entire network and maintain these attacks as undetected. Finally, an event-triggered based detector is proposed that can be used to detect the quasi-covert cyber attacks. Numerical simulations are provided to illustrate the effectiveness and capabilities of our proposed methodologies.
\end{abstract}

\section{Introduction}
Multi-agent systems (MAS) have wide range of applications as in smart grids, transportation systems, internet of things (IoT), and unmanned aerial vehicles (UAV). Due to their wide range of applications these cyber systems have been studied during the past two decades \cite{1333204,4804653,dfdfi}. One of the main problems in MAS is to reach consensus in a distributed manner. This problem has been addressed in the literature for various types of agents. In \cite{zhang2011optimal,li2009consensus}, observer-based consensus protocols have been proposed. Distributed consensus protocols require agents to transmit their information to the neighboring agents through a communication network. This communication network is prone and has vulnerabilities to cyber attacks.

In recent years, addressing security related problems of cyber-physical systems (CPS) has been a challenging issue \cite{adaii,ascf,siolcps,9029298,mustafaa,dibaji}. In \cite{adaii,ascf}, frameworks to study various types of cyber attacks on linear time-invariant (LTI) systems have been proposed. The impact of a certain type of cyber attack on MAS in which the adversary uses the model of the system to generate its attack signals has been studied in \cite{mustafaa}. It is shown if the root of the directed spanning tree contained in the network graph is under ``cyber-physical" attacks, the entire MAS can become unstable. In \cite{boem2017distributed}, a distributed methodology to detect cyber attacks on communication networks among interconnected systems and MAS that are equipped with consensus-based controllers has been proposed. However, conditions under which the adversary is capable of performing undetectable cyber attacks in MAS \underline{has not been} investigated in above references.

To increase security and reduce consumption of energy, one can employ an event-triggered protocol so  agents in MAS can communicate with one another. In \cite{zhang2013observer,yang2016decentralized,davoodiEvent}, various types of event-triggered observer-based methodologies have been proposed for linear MAS and LTI systems. The work in \cite{HajiEvent} has studied an event-triggered unit that can be used to simultaneously reach a consensus and detect faults in MAS.

In this paper, our goal is to investigate conditions under which cyber attacks on communication links of multi-agent CPS become undetectable. The MAS under having a directed communication graph is equipped with consensus-based observers that are under \textit{false data injection attacks}. Our first objective is to show that if an adversary attacks the root of the directed spanning tree of the communication graph with a certain attack signal, all the agents follow that attack signal and reach a new consensus point. This implies that one is injecting and introducing an undetectable cyber attack in the sense that residual signals in presence of cyber attacks approach to zero as time approaches to infinity.

Next, we investigate detectability conditions of cyber attacks on various teams of agents in case that non-root agents are under cyber attacks. Since attacking a non-root agent can become detectable to a certain group of agents, a cyber attack methodology that is denoted as ``\textit{quasi-covert cyber attack}" is introduced. Through this novel cyber attack, the adversary can use it to eliminate the impacts of cyber attacks on agents that can detect cyber attacks. Consequently, by performing quasi-covert attacks on non-root agents, the adversary is capable of concealing his/her attacks and making them undetectable. Finally, an event-triggered detector is proposed in this work that under certain conditions can be used to generate residuals that are sensitive to cyber attack signals \underline{even} in presence of injection of quasi-covert cyber attacks.

The main contributions of this paper are threefold. First, a definition is introduced and proposed that specifies characteristics of undetectable cyber attacks on MAS. Then conditions  on the graph topology and its Laplacian matrix along with detectors of MAS are developed so that an adversary is capable of performing undetectable cyber attacks. Moreover, if the above does not hold, we investigate under what conditions cyber attacks are detectable on a certain team of agents. Second, quasi-covert cyber attacks are introduced where malicious hackers can inject in order to maintain their attacks undetected provided only non-root agents are compromised. Finally, an event-triggered detector is proposed for quasi-covert cyber attacks that given its event-based communication strategy is more secure in comparison with conventional communication protocols.

\section{Preliminaries}\label{s:pre}

\subsection{Graph Theory}
A directed graph (digraph) $\mathcal{G}$ is defined with a set of vertices or nodes $\mathcal{V}=\{1, \, 2, \dots, \, N\}$ and the set $\mathcal{E} \subset \mathcal{V} \times \mathcal{V}$ that denotes the edges of the digraph. The pair of distinct vertices $\mathcal{G}: (i,j) \in \mathcal{E}$ defines an edge. Graph $\mathcal{G}$ is called directed if $(i,j) \in \mathcal{E}$ does not imply $(j,i) \in \mathcal{E}$. The matrix $\mathcal{A}=[a_{ij}] \in \mathbb{R}^{N \times N}$ denotes the adjacency matrix of $\mathcal{G}$, where $a_{ij}=1$ when there exists a link from node $j$ to $i$. The set $\mathcal{N}_i$ denotes the set of neighbors of the node $i$ which contains those nodes that have an edge to $i$. Moreover, $|\mathcal{N}_i|=d_i$, where $|\cdot|$ is the cardinality of the set. The Laplacian matrix of the graph $\mathcal{G}$ is defined as $L=D-\mathcal{A}$, where $D=\text{diag}(d_1, \, d_2, \dots, \, d_N)$ denotes the in-degree matrix. A directed path between nodes $i$ and $j$, that is denoted by $\mathcal{P}_{ij}$, is a sequence of edges that connects the node $i$ to the node $j$ and follows along the direction of edges.

\subsection{Model of MAS and State Estimation Module}
The dynamics of MAS consisting of $N$ agents is given by

\begin{equation}\label{e:agent_i}
\begin{split}
\dot{x}_i (t)&=A x_i (t)+Bu_i(t), \\
y_i(t)&= C x_i (t), \, \, \, \, i=1,\dots,N,
\end{split}
\end{equation}

\noindent where $x_i (t) \in \mathbb{R}^n$ denotes the state of the $i$-th agent , $u_i (t) \in \mathbb{R}^m$ denotes the control input of that agent, and $y_i (t) \in \mathbb{R}^p$ represents the output measurement of agent $i$. The known matrices $(A, \, B, \, C)$ are of appropriate dimensions and each agent is assumed to be controllable and observable.

Since only few states of agents are assumed to be measurable ($p <n$) one needs to design an observer-based controller. In \cite{zhang2011optimal}, for MAS in the form of \eqref{e:agent_i} an observer-based consensus protocol was proposed as follows:

\begin{equation}\label{e:obs_i}
\begin{split}
\dot{\hat{x}}_i (t) =& A \hat{x}_i (t)+ B u_i(t)-cF\zeta_i(t),\\
u_i (t)  =&  cK \epsilon_i (t),
\end{split}
\end{equation}

\noindent where $\hat{x}_i(t) \in \mathbb{R}^n$ denotes the observer state of the $i$-th agent, $\zeta_i(t)= \sum_{j \in \mathcal{N}_i}((y_j (t)-y_i (t))+C(\hat{x}_i (t) -\hat{x}_j (t)))$, $\epsilon_i(t)= \sum_{j \in \mathcal{N}_i}(\hat{x}_i (t) -\hat{x}_j (t))$, $c \in \mathbb{R}$ is a scalar, $F \in \mathbb{R}^{n \times p}$ is the observer gain matrix, and $K \in \mathbb{R}^{m\times n}$ denotes the control gain matrix to be selected. It should be noted that the observer \eqref{e:obs_i} is a special case of that in \cite{zhang2011optimal} in which the MAS do not have a leader.

\begin{assumption}\label{assu:spanning}
	The directed graph $\mathcal{G}$ contains at least one directed spanning tree. The set $V_\text{r}=\{i_\text{r}, \, i_{\text{r}}+1, \dots, i_{\text{r}}+N_\text{r}-1 \}$ contains agents that constitute as roots of directed spanning trees in $\mathcal{G}$ and $N_\text{r}$ denotes the number of these agents.
\end{assumption}

\begin{lem}[\cite{ren2010distributed,yang2016decentralized}]\label{lem:tree}
    The algebraic multiplicity of the eigenvalue $\lambda_0=0$ of the Laplacian matrix $L_0$ that belongs to the graph $\mathcal{G}_0$ is one if $\mathcal{G}_0$ contains a directed spanning tree. Furthermore, the other eigenvalues have positive real parts and the right and left eigenvectors associated with $\lambda_0$ are denoted by $\mathbf{1}_N$ and $r_0=[r_1, \dots, r_N]^\top \in \mathbb{R}^N$, respectively, where $\sum_{j=1}^{N}r_j=1$ with scalar $r_j$.
\end{lem}

\begin{definition}\label{def:consensus}
	For the MAS \eqref{e:agent_i}, consensus is achieved if for any initial condition
		$$\lim_{t \rightarrow \infty} \|y_i(t) -y_j(t)\|=0, \, \lim_{t \rightarrow \infty} \|\hat{x}_i(t) -\hat{x}_j(t)\|=0,$$
	
	\noindent $\forall \, i,j=1,\dots,N$, where $\| . \|$ denotes the Euclidean norm.
\end{definition}

It has been shown in \cite{zhang2011optimal} that if the Assumption \ref{assu:spanning} holds, the observer and control parameters in \eqref{e:obs_i} can be designed such that the closed-loop system reaches a consensus in  sense of the Definition \ref{def:consensus}.

\begin{assumption}\label{assu:observer_design}
	Parameters of the observer-based consensus protocol \eqref{e:obs_i} are designed such that  conditions in Theorem 4 in \cite{zhang2011optimal} are satisfied. Hence, the closed-loop system \eqref{e:agent_i} and \eqref{e:obs_i} reach a consensus under cyber attack free conditions.
\end{assumption}

\begin{lem}[\cite{yang2016decentralized}]\label{lem:event}
	Given the Hurwitz matrix $H \in \mathbb{R}^{n \times n}$, there exists a nonsingular matrix $P_\text{H}$ that satisfies $P_\text{H}^{-1} H P_\text{H}=J_\text{H}$ and $\|e^{Ht} \| \leq \|P_\text{H}\| \|P_\text{H}^{-1} \| c_\text{H}e^{\lambda_\text{H}^\text{m} t}$, $\forall \, t \geq 0$, where $J_\text{H}$ denotes the Jordan canonical form of $H$, $c_\text{H} >0$ is a positive constant, and $\text{max} \, \text{Re}(\lambda(H)) < \lambda_\text{H}^\text{m} <0$.
\end{lem}

\section{Problem Formulation}\label{s:pf}

\subsection{Cyber Attacks on MAS and Residual Generation}
The state estimator \eqref{e:obs_i} for the $i$-th agent receives information from the agent $j \in \mathcal{N}_i$ that is represented as pair $p_{ji}(t)=(\hat{x}_j (t), y_j (t))$. Since the communication links among agents are vulnerable to cyber attacks, the adversary is capable of injecting attack signals $a_{\hat{x}}^{ji}(t) \in \mathbb{R}^n$ and $a_{y}^{ji}(t) \in \mathbb{R}^p$ to $p_{ji}(t)$. Therefore, the agent $i$ under cyber attacks receives the manipulated pair $p^a_{ji}(t)=(\hat{x}_j (t)+a_{\hat{x}}^{ji}(t), y_j (t)+a_{y}^{ji}(t))$ from the agent $j$.

\begin{assumption}\label{assu:allattacked}
	We consider the worst case scenario cyber attack in which case all the incoming communication links of a given agent are being compromised.
\end{assumption}

\begin{remark}
	Cyber attacks on MAS in this paper can be considered as the ``man-in-the-middle" type of attack \cite{man}. In this cyber attack type the adversary first blocks the received information from a group of neighboring agents of a compromised agent. Following this the adversaries inject their attack signals as legitimate information that are then received from a group of neighboring agents and transmitted to the targeted agent.
\end{remark}

Consequently, the closed-loop equations of the MAS \eqref{e:agent_i} and the observer \eqref{e:obs_i} under cyber attacks can be represented in the following form:

\begin{IEEEeqnarray}{rCl}
	\dot{x}_i (t)&=& A x_i (t)+BcK \epsilon_i (t)-BcK \sum_{j \in \mathcal{N}_i} q_{i}a_{\hat{x}}^{ji}(t), \label{e:agent_i_attack} \\
	\dot{\hat{x}}_i (t) & = & A \hat{x}_i (t) + B cK \epsilon_i (t)-cF\zeta_i(t)-BcK \sum_{j \in \mathcal{N}_i} q_{i}a_{\hat{x}}^{ji}(t) \nonumber \\
	&& -cF\sum_{j \in \mathcal{N}_i}q_{i}(a_{y}^{ji}(t)-Ca_{\hat{x}}^{ji}(t)), \, \, \, \, i=1,\dots,N, \label{e:obs_i_attack}
\end{IEEEeqnarray}

\noindent where $q_{i}=1$ if the incoming communication links to the agent $i$ are compromised and under cyber attack, and $q_{i}=0$, otherwise.

To detect anomalies in the $i$-th agent or its neighboring agents one can utilize the received information from the neighboring agents and generate the following residual signals:
\begin{IEEEeqnarray}{rCl}
	res^i_y(t)&=& \sum_{j \in \mathcal{N}_i} \|(y_j (t)-y_i (t))+q_{i}a_{y}^{ji}(t) \|, \label{e:res_i_y} \\
	res^i_{\hat{x}} (t)&=& \sum_{j \in \mathcal{N}_i} \|(\hat{x}_i (t)-\hat{x}_j (t))-q_{i}a_{\hat{x}}^{ji}(t) \|. \label{e:res_i_x}
\end{IEEEeqnarray}

\subsection{Objectives}
We are pursuing  three main objectives in this paper. First, we introduce a formal definition for undetectable cyber attacks on MAS in the sense that residuals \eqref{e:res_i_y} and \eqref{e:res_i_x} converge to zero as time approaches to infinity. Following that we show that under certain conditions on the network the cyber attacks on a given agent which happens to be the root of the communication graph can become undetectable. The second objective is to introduce a novel type of undetectable cyber attacks that are designated as \textit{quasi-covert cyber attacks} in which if a group of non-root agents are compromised, still the adversary is capable of eliminating and hiding the impact of cyber attacks on agents that could otherwise detect them. The final objective is to develop an event-triggered based detector with the goal of detecting quasi-covert cyber attack signals that are developed in the second objective.

\section{Undetectable Cyber Attacks}\label{s:undetectable}
In this section, conditions to introduce and inject undetectable cyber attacks are investigated and developed. Let us define $\check{A}=I_2 \otimes A$, $\check{K}=I_2 \otimes cK$, $\check{C}=I_2 \otimes C$, and

\begin{equation}\label{e:def1}
\begin{split}
\check{B}=& \begin{bmatrix}
0 & B \\
0 & B
\end{bmatrix}, \,
\check{F} =\begin{bmatrix}
0 & 0 \\
cF & -cF
\end{bmatrix}, \\
 \check{B}_\text{a} =& \begin{bmatrix}
-BcK & 0 \\
-BcK-cFC & cF
\end{bmatrix}.
\end{split}
\end{equation}

\noindent By utilizing \eqref{e:def1} and augmenting the dynamics \eqref{e:agent_i_attack} and \eqref{e:obs_i_attack} yields:
\begin{equation}\label{e:sys_aug_1}
\begin{split}
\dot{\check{x}}_i (t) =& \check{A} \check{x}_i(t)+ (\check{B} \check{K} - \check{F} \check{C}) \sum_{j \in \mathcal{N}_i}(\check{x}_i (t)-\check{x}_j (t)) \\
 &+\check{B}_\text{a} \check{a}_{i} (t),
\end{split}
\end{equation}

\noindent where $\check{x}_i(t)= [x_i(t)^\top ,\, \hat{x}_i(t)^\top]^\top$, $\check{a}_{i} (t)= \sum_{j \in \mathcal{N}_i} q_{i} \check{a}_{ji}$, and $\check{a}_{ji}(t)=[a_{\hat{x}}^{ji}(t)^\top , \, a_{y}^{ji}(t)^\top]^\top$.
\begin{definition}\label{def:undetectable}
	The cyber attacks on communication links in the MAS \eqref{e:sys_aug_1} are undetectable by the entire network if the MAS reaches a consensus under attack free conditions, i.e., $\check{a}_{i} (t)=0$, according to the Definition \ref{def:consensus}, and in presence of cyber attacks, i.e., $\check{a}_{i} (t) \neq 0$, $\forall \, t>0$, the following equations are satisfied:
	\begin{IEEEeqnarray}{rCl}
		\lim_{t \rightarrow \infty} \|y_i(t) -y_j(t)-q_{i}a_{y}^{ji}(t)\|=0, \label{e:cond_undetec_1} \\
		\lim_{t \rightarrow \infty} \|\hat{x}_i(t) -\hat{x}_j(t)-q_{i}a_{\hat{x}}^{ji}(t)\|=0, \label{e:cond_undetec_2}
	\end{IEEEeqnarray}

	\noindent $\forall \, i,j=1,\dots,N$.
\end{definition}

\begin{remark}
	If $\check{a}_{i} (t)=0$ and the MAS reaches a consensus, the residuals \eqref{e:res_i_y} and \eqref{e:res_i_x} converge to zero as $t \rightarrow \infty$ as expected since no cyber attack is injected to the system. If in presence of cyber attacks, i.e., $\check{a}_{i} (t) \neq 0$, equations \eqref{e:cond_undetec_1} and \eqref{e:cond_undetec_2} are satisfied, it can be inferred that $\lim_{t \rightarrow \infty} res^i_y(t)=0$ and $\lim_{t \rightarrow \infty} res^i_{\hat{x}} (t)=0$ for $i=1,\dots,N$. Consequently, the cyber attack is undetectable.
\end{remark}

\begin{theorem}\label{th:root_attack}
	Under Assumptions \ref{assu:spanning}-\ref{assu:allattacked}, let the cyber attacks on the agent $i$ in \eqref{e:sys_aug_1} be selected as $a_{\hat{x}}^{ji}(t)={a}_0- \hat{x}_j (t)$ and $a_{y}^{ji}(t)=C {a}_0- y_j (t)$ for $j \in \mathcal{N}_i$, where ${a}_0 \in \mathbb{R}^n$ is a constant vector. The above cyber attacks are undetectable by all  agents if $i \in V_\text{r}$, and $A+BcK$ and $A+\tilde{\lambda}_k^\text{r} cFC$ for $k=2, \dots, \, N+1$ are Hurwitz, where $\tilde{\lambda}_k^\text{r} \neq 0$ is the $k$-th eigenvalue of the matrix $$\tilde{L}_\text{r}= \begin{bmatrix}
	0 & \mathbf{0}_{1 \times N} \\
	-D_\text{r} & {L}_\text{r} \\
	\end{bmatrix},$$
	where $D_\text{r}=[\mathbf{0}_{i-1}^\top, \, d_{i},\, \mathbf{0}_{N-i}^\top]^\top$, ${L}_\text{r}=L+Q_\text{r}\mathcal{A}$, $Q_\text{r}= \text{diag}(\mathbf{0}_{i-1},q_i,\mathbf{0}_{N-i})$, and $\mathbf{0}_{m \times n} \in \mathbb{R}^{m \times n}$ denotes a matrix with all entries equal to 0.
\end{theorem}

\begin{proof}
	By adding the cyber attack signals to \eqref{e:sys_aug_1} and augmenting the dynamics of agents, one obtains
	\begin{equation*}
	\begin{split}
	\dot{\check{x}} (t)=(I_N \otimes \check{A}+{L}_\text{r} \otimes (\check{B} \check{K}- \check{F} \check{C})) \check{x}(t)+ D_\text{r} \otimes \check{B} \check{a}_0,
	\end{split}
	\end{equation*}
	
	\noindent where $\check{x}(t)=[\check{x}_1(t)^\top , \, \check{x}_2(t)^\top, \dots, \, \check{x}_N(t)^\top]^\top$, $\check{a}_0 = [a_0^\top , (C a_0)^\top]^\top$. One can consider the adversary as a virtual agent that transmits its information $a_0$ and $C a_0$ to the $i$-th agent. Since from the agent $i$ there exists paths to the rest of $N-1$ agents the virtual adversary agent is the root of a directed spanning tree that is contained in the graph $\mathcal{G}_\text{r}$ with $\tilde{L}_\text{r}$ as its Laplacian matrix. Hence, the state space representation of the MAS augmented with the virtual adversary agent can be represented as follows:
	\begin{equation*}
	\begin{split}
	\dot{\tilde{x}}_\text{r} (t)= (\tilde{I}_\text{r} \otimes \check{A}+\tilde{L}_\text{r} \otimes (\check{B} \check{K} - \check{F} \check{C})) \tilde{x}_\text{r} (t),
	\end{split}
	\end{equation*}
	
	\noindent where $\tilde{x}_\text{r} (t)=[\mathbf{1}_2^\top \otimes {a}_0^\top , \, \check{x}(t)^\top]^\top$, $\tilde{I}_\text{r}=\begin{bmatrix}
	0 & \mathbf{0}_{1 \times N} \\
	\mathbf{0}_{N } & I_{N}
	\end{bmatrix}$, and $\mathbf{1}_n \in \mathbb{R}^n$ denotes a vector with all its elements equal to $1$. In similar manner as in \cite{li2009consensus}, the disagreement vector $\delta_\text{r}(t)=\tilde{x}_\text{r}(t)-(\mathbf{1}_{N+1} \tilde{r}_\text{r}^\top \otimes I_{2n})\tilde{x}_\text{r}(t)$ is defined, where $\tilde{r}_\text{r}$ is the left eigenvector of $\tilde{L}_\text{r}$ as defined in the Lemma \ref{lem:tree}. Consequently, we obtain
	\begin{equation}\label{e:delta_r}
	\dot{\delta}_r(t)=(\tilde{I}_\text{r} \otimes \check{A}+\tilde{L}_\text{r} \otimes (\check{B} \check{K} - \check{F} \check{C}))\delta_\text{r} (t).
	\end{equation}
	
	\noindent It can be shown from the definition of $\delta_r (t)$ that $\delta_r (t)=0$ if and only if $\mathbf{1}_2 \otimes {a}_0= \check{x}_{1}(t)= \, \check{x}_{2}(t)= \dots = \, \check{x}_{N}(t)$. Therefore, if \eqref{e:delta_r} is stable, one can conclude that all the agents reach ${a}_0$ corresponding to the consensus set point.

	From Lemma \ref{lem:tree} and considering that $\mathcal{G}_r$ contains a directed spanning tree it follows that there exist matrices $T \in \mathbb{R}^{ N+1 \times N+1}$, $Y \in \mathbb{R}^{N+1 \times N}$, $W \in \mathbb{R}^{N \times N+1}$, and block diagonal matrix $\Delta \in \mathbb{R}^{N \times N}$ with diagonal entries equal to nonzero eigenvalues of $\tilde{L}_\text{r}$ such that \cite{li2009consensus}:
	\begin{equation*}
	T=[\mathbf{1}_{N_{}+1} \, Y], \, T^{-1}= \begin{bmatrix}
	\tilde{r}^\top \\
	W
	\end{bmatrix}, \, T^{-1} \tilde{L} T= J=\begin{bmatrix}
	0 & \mathbf{0}_{1 \times N} \\
	\mathbf{0}_{N} & \Delta
	\end{bmatrix}.
	\end{equation*}
	
	Let us define $\varepsilon (t)=(T^{-1} \otimes I_{2n})\delta_\text{r} (t)=[\varepsilon_1 (t)^\top, \, \varepsilon_{2:N+1} (t)^\top]^\top$. It follows from Lemma \ref{lem:tree} and definition of $\delta_\text{r} (t)$ that $	\varepsilon_1 (t) = \mathbf{0}_{2n}$ and
	\begin{equation}
	\dot{\varepsilon}_{2:N+1} (t) = (I_{N} \otimes \check{A}+\Delta \otimes (\check{B} \check{K} - \check{F} \check{C}))\varepsilon_{2:N+1} (t).
	\end{equation}
	
	\noindent It is easy to show that the matrices $I_{N} \otimes \check{A}+\tilde{\lambda}_k^\text{r} \otimes (\check{B} \check{K} - \check{F} \check{C})$ for $k=2, \dots, \, N+1$ having a similar structure to $\begin{bmatrix}
	A+\tilde{\lambda}_k^\text{r} cFC & \mathbf{0}_{n \times n} \\
	-\tilde{\lambda}_k^\text{r} cFC & A+BcK
	\end{bmatrix}.$
	
	Consequently, if $A+BcK$ and $A+\tilde{\lambda}_k^\text{r} cFC$ for $k=2, \dots, \, N+1$ are Hurwitz, \eqref{e:delta_r} is stable and states of all the agents reach $a_0$ as $t \rightarrow \infty$. Hence, according to the Definition \ref{def:undetectable}, the cyber attacks are undetectable from the entire network.
\end{proof}

\subsection{Cyber Attacks Injected to Non-Root Agents}
In this subsection, cyber attacks that are injected to the communication links of agents that do not belong to the set $V_\text{r}$ are investigated.

\begin{definition}\label{def:directly_attacked}
	The agents in the network that are directly under cyber attacks and their incoming communication channels are compromised are included in the set $V_\text{da}=\{i_\text{a}, \, i_{\text{a}}+1, \dots, \, i_{\text{a}}+N_\text{da}-1 \}$, where $N_\text{da}$ is the number of directly attacked agents.
\end{definition}

\begin{definition}[\cite{ugrinovskii2013conditions}]\label{def:reachable}
	The reachable subgraph of the vertex $i$, $R(i)$, is defined as the vertex subgraph that contains the node $i$ and all reachable nodes from it.
\end{definition}

\begin{definition}\label{def:healthy}
	The directed path between the vertices $i$ and $j$, denoted by $\mathcal{P}_{ij}$, is a directed healthy path if none of the communication links on this path is compromised by adversaries.
\end{definition}

By utilizing the Definitions \ref{def:directly_attacked}-\ref{def:healthy}, agents in the network can be partitioned into two groups. For agents from the root nodes of the directed tree, i.e., the set $V_\text{r}$, there exist directed healthy paths that constitute the first group. Agents in this group are designated as ``\textit{uncompromised}" where their states can be stacked into the vector $x_\text{nc} (t)=[\check{x}_1 (t)^\top, \, \check{x}_2 (t)^\top, \dots , \, \check{x}_{N_\text{nc}} (t)^\top]^\top$. The second group consists of the remaining network agents. This group contains agents in the vertex subgraphs of directly attacked agents, i.e., the set $V_\text{da}$, such that there does not exist any directed healthy path among agents in $V_\text{r}$ and these nodes. The agents that belong to this group are designated as ``\textit{attacked}" agents and $x_\text{a}(t) = [\check{x}_{N_\text{nc}+1}(t)^\top, \, \check{x}_{N_\text{nc}+2}(t)^\top, \dots , \, \check{x}_{N}(t)^\top]^\top$ represents the states of this group. The subscripts ``\textit{nc}" and ``\textit{a}" are employed to denote the $N_\text{nc}$ uncompromised agents and the $N_\text{a}$ attacked agents, respectively. Without loss of any generality, we assume that the first $N_\text{nc}$ agents are not under cyber attacks. Consequently, the Laplacian matrix is partitioned into the following form:
\begin{equation}
L=\begin{bmatrix}
L_\text{nc} & l_\text{nca} \\
l_\text{anc} & L_\text{a}
\end{bmatrix},
\end{equation}

\noindent where $L_\text{nc}\in \mathbb{R}^{N_\text{nc} \times N_\text{nc}}$, $L_\text{a} \in \mathbb{R}^{N_\text{a} \times N_\text{a}}$, $l_\text{nca} \in \mathbb{R}^{N_\text{nc} \times N_\text{a}}$, and $l_\text{anc} \in \mathbb{R}^{N_\text{a} \times N_\text{nc}}$.

The state space representation of the entire MAS can now be described in the following form:
\begin{IEEEeqnarray}{rCl}
	\dot{x}_\text{a} (t) &=& A_\text{a} x_\text{a} (t)+A_\text{anc}x_\text{nc} (t)+B_\text{a} a(t), \label{e:MAS_a} \\
	\dot{x}_\text{nc} (t) &=& A_\text{nc} x_\text{nc} (t)+A_\text{nca}x_\text{a} (t) , \label{e:MAS_nc}
\end{IEEEeqnarray}

\noindent where $A_\text{a}=I_{N_\text{a}} \otimes \check{A}+L_\text{a} \otimes (\check{B} \check{K}-\check{F} \check{C})$, $A_\text{anc}=l_\text{anc}\otimes (\check{B} \check{K}-\check{F} \check{C})$, $A_\text{nc}= I_{N_\text{nc}} \otimes \check{A}+L_\text{nc} \otimes (\check{B} \check{K}-\check{F} \check{C})$, $A_\text{nca}= l_\text{nca} \otimes (\check{B} \check{K}-\check{F} \check{C})$, $B_\text{a}= Q_\text{a} \otimes \check{B}_\text{a}$, and $Q_\text{a}=\text{diag}(q_{N_\text{nc}+1}, \, q_{N_\text{nc}+2}, \dots, \, q_{N}) \in \mathbb{R}^{N_\text{a} \times N_\text{a}}$ is a diagonal matrix.

\begin{definition}\label{def:agent_groups}
The set of $N$ agents is partitioned into two subsets, namely $V_\text{nc}= \{1, 2, \dots, \, N_{\text{nc}} \}$ and $V_\text{a}= \{N_{\text{nc}}+1, \, N_{\text{nc}}+2, \dots, \, N \}$, that contain uncompromised and attacked agents, respectively.
\end{definition}

\begin{definition}\label{def:nca}
	The set of uncompromised agents that receive information from the set of attacked agents are defined by $V_\text{nca}=\{i_\text{nca}, \, i_{\text{nca}}+1, \dots , \, i_\text{nca}+N_\text{nca}-1 \}$, where $N_\text{nca}$ is the number of uncompromised agents that receive information from the agents in $V_\text{a}$. In other words, $j \in V_\text{nca}$ if and only if $j \in V_\text{nc}$ and there exists $i \in V_\text{a}$ such that $i \in \mathcal{N}_j$.
\end{definition}

\begin{assumption}\label{assu:not_root}
	None of agents in the set $V_\text{r}$ is directly under attack, i.e., $V_\text{r} \cap V_\text{da}=0$.
\end{assumption}

Without loss of generality, let us assume that the first $N_\text{da}$ agents in the set $V_\text{a}$ create the set of directly attacked agents. We are now in a position to state our main result of this subsection.
\begin{theorem}\label{th:detectable}
	Let the Assumptions \ref{assu:spanning}-\ref{assu:not_root} hold and the cyber attacks on agents in the set $V_\text{da}$ are designed as $a_{\hat{x}}^{ji}(t)={a}_0- \hat{x}_j (t)$ and $a_{y}^{ji}(t)=C {a}_0- y_j (t)$ for $j=1, \dots , \, N_\text{nc}$ and $j \in \mathcal{N}_i$, where $C {a}_0 \neq 	 \lim_{t \rightarrow \infty} y_{i_\text{r}} (t)$, $ {a}_0 \neq \lim_{t \rightarrow \infty} x_{i_\text{r}} (t)$, and ${a}_0 \in \mathbb{R}^n$ is a constant vector. Consequently, the following can be stated:
	\begin{enumerate}
		\item The cyber attacks are undetectable on the set of nodes $V_\text{a}$ if $A+BcK$ and $A+\tilde{\lambda}_k cFC$ for $k=2, \dots, \, N_\text{a}+1$ are Hurwitz, where $\tilde{\lambda}_k \neq 0$ denotes the $k$-th eigenvalue of the Laplacian matrix
		$$\tilde{L}_\text{a}= \begin{bmatrix}
		0 & \mathbf{0}_{1 \times N_{a}} \\
		-D_\text{a} & L_\text{a} \\
		\end{bmatrix},$$
		with $D_\text{a}=[d_{i_\text{a}}^\text{nc}, \, d_{i_\text{a}+1}^\text{nc}, \dots, \, d_{i_{\text{a}+N_\text{da}-1}}^\text{nc}, \, \mathbf{0}_{N_{a}-N_{da}}^\top]^\top$ and $d_{i}^\text{nc}= |\mathcal{N}_{i} \cap V_\text{nc}|$ for $i \in V_\text{da}$.
		
		\item The cyber attacks are detectable by agents that belong to the set $V_\text{nca}$.
		
	\end{enumerate}
\end{theorem}

\begin{proof}
	Substituting the cyber attack signals into \eqref{e:MAS_a} the following can be expressed:
	\begin{equation}\label{e:th_x_a}
	\begin{split}
	\dot{x}_\text{a} (t)= (I_{N_\text{a}} \otimes \check{A}+L_\text{a} \otimes (\check{B} \check{K} - \check{F} \check{C})) x_\text{a} (t)+ D_\text{a} \otimes \check{B}_\text{a} \check{a}_0 ,
	\end{split}
	\end{equation}
	
	\noindent where $\check{a}_0 = [a_0^\top , ( C a_0)^\top]^\top$. It follows from \eqref{e:th_x_a} that the nodes which belong to the set $V_\text{a}$ do not receive information from the nodes in the set $V_\text{nc}$. Since the cyber attack signal $a_0$ is the same for all agents in the set $V_\text{da}$, one can consider the adversary as a virtual agent that is the root of a directed spanning tree contained in the graph $\mathcal{G}_\text{a}$ with $\tilde{L}_\text{a}$ as its Laplacian matrix. The virtual adversary agent transmits its information $a_0$ and $C a_0$ to all the directly attacked agents. Consequently, the dynamics of the attacked agents augmented with the virtual adversary agent can be derived as given below:
\begin{equation*}
\begin{split}
\dot{\tilde{x}}_\text{a} (t)= (\tilde{I}_\text{a} \otimes \check{A}+\tilde{L}_\text{a} \otimes (\check{B} \check{K} - \check{F} \check{C})) \tilde{x}_\text{a} (t),
\end{split}
\end{equation*}

\noindent where $\tilde{x}_\text{a} (t)=[\mathbf{1}_2^\top \otimes {a}_0^\top , \, x_\text{a}(t)^\top]^\top$ and $\tilde{I}_\text{a}=\begin{bmatrix}
0 & \mathbf{0}_{1 \times N_{a}} \\
\mathbf{0}_{N_{a} } & I_{N_\text{a}}
\end{bmatrix}$.

Following along the same steps as in the proof of the Theorem \ref{th:root_attack} it can be concluded that if $A+BcK$ and $A+\tilde{\lambda}_k cFC$, with $k=2, \dots, \, N_\text{a}+1$, are Hurwitz, the states of all agents in the set $V_\text{a}$ reach $a_0$ as $t \rightarrow \infty$. Therefore, according to the Definition \ref{def:undetectable} the cyber attacks are undetectable by agents that belong to the set $V_\text{a}$.

Suppose all agents in \eqref{e:MAS_a} and \eqref{e:MAS_nc} reach a consensus. Following along the results  previously derived it follows that states of the attacked agents, and consequently uncompromised agents, should reach $a_0$, which contradicts the assumption that $C {a}_0 \neq 	\lim_{t \rightarrow \infty} y_{i_\text{r}} (t)$ and $ {a}_0 \neq \lim_{t \rightarrow \infty} x_{i_\text{r}} (t)$. Therefore, the entire network cannot reach a consensus. The residuals for the agent $j \in V_\text{nc}$ with $i \in V_\text{a}$ and $i \in \mathcal{N}_j$ are nonzero as $t \rightarrow \infty$ since $\hat{x}_i (t) \neq \hat{x}_j (t)$ and $y_i (t) \neq y_j (t)$. This completes the proof of the theorem.
\end{proof}

\begin{remark}
	Suppose in Theorem \ref{th:detectable} one has $C {a}_0 = \lim_{t \rightarrow \infty} y_{i_\text{r}} (t)$ and $ {a}_0 = \lim_{t \rightarrow \infty} x_{i_\text{r}} (t)$, which implies that the attack signals will not change the consensus set point and follow the control objective of the MAS. Hence, it is reasonable to assume that the control objective of adversaries differs from the group of agents, in other words $C {a}_0 \neq 	 \lim_{t \rightarrow \infty} y_{i_\text{r}} (t)$ and $ {a}_0 \neq \lim_{t \rightarrow \infty} x_{i_\text{r}} (t)$.
\end{remark}

\subsection{Quasi-Covert Cyber Attack on the MAS Network}\label{sub_s:covert}

In the Theorem \ref{th:detectable}, it was shown that  cyber attack signals are detectable by agents that belong to the set $V_\text{nc}$ and receive information from agents in the set $V_\text{a}$. Hence, if an adversary attacks the communication channels that connect agents in these two sets, they may be able to manipulate the transmitted information such that impacts of cyber attacks are made hidden and eliminated. In this paper, this attack methodology where impacts of cyber attacks on the set $V_\text{nc}$ are eliminated is denoted by the ``\textit{quasi-covert cyber attacks}".

\begin{assumption}\label{assu:knowsTheParameters}
	The adversary has knowledge on the parameters of the MAS \eqref{e:agent_i} and the observer-based consensus protocol \eqref{e:obs_i}.
\end{assumption}

To conceal impacts of cyber attacks on the set $V_\text{nc}$ the adversary needs to attack the outgoing communication links from the nodes in the set $V_\text{a}$ to the agents that belong to $V_\text{nc}$. These agents are included in the set $V_\text{nca}$. To achieve this, under Assumption \ref{assu:knowsTheParameters} the adversary is capable of running the following process:
\begin{IEEEeqnarray}{rCl}
	\dot{x}^\text{c}_i (t) &=& A x^\text{c}_i (t)+BcK \epsilon_i^\text{c}(t) , \label{e:agent_covert} \\
	\dot{\hat{x}}^\text{c}_i (t) &=& A \hat{x}^\text{c}_i (t)+BcK \epsilon_i^\text{c}-cF \zeta_i^\text{c} (t) , \label{e:obs_covert}
\end{IEEEeqnarray}

\noindent for $i = N_{\text{nc}}+1, \, N_{\text{nc}}+2, \dots, \, N$, where $x^\text{c}_i (t), \, \hat{x}^\text{c}_i (t)  \in \mathbb{R}^n$, for any arbitrary initial conditions $x^\text{c}_i (0)$ and $\hat{x}^\text{c}_i (0)$, and
\begin{equation*}
\begin{split}
\epsilon_i^\text{c}(t)=& \sum_{k \in \mathcal{N}_i \cap V_\text{a}}(\hat{x}_i^\text{c} (t)- \hat{x}_k^\text{c} (t)) +\sum_{j \in \mathcal{N}_i \cap V_\text{nc}}(\hat{x}_i^\text{c} (t) - \hat{x}_j (t)), \\
\zeta_i^\text{c} (t)= & \sum_{k \in \mathcal{N}_i \cap V_\text{a}}((C x_k^\text{c} (t) - C x_i^{c} (t))+ C(\hat{x}^\text{c}_i-\hat{x}^\text{c}_k (t))) \\
& + \sum_{j \in \mathcal{N}_i \cap V_\text{nc}}(( y_j (t) - C x_i^{c} (t))+ C(\hat{x}^\text{c}_i-\hat{x}_j (t))).
\end{split}
\end{equation*}

\begin{lem}\label{lem:quasi_covert}
	Let the Assumptions \ref{assu:spanning}-\ref{assu:knowsTheParameters} hold and agents in the set $V_\text{da}$ are under cyber attacks as specified in the Theorem \ref{th:detectable}. The adversary is capable  of eliminating impacts of these cyber attacks and make them undetectable on the set $V_\text{nca}$ by adding cyber attack signals $\hat{a}_{\hat{x}}^{ij} (t)= \hat{x}^\text{c}_i (t) - \hat{x}_i (t)$ and $\hat{a}_y^{ij} (t)=C x^\text{c}_i (t)-y_i (t)$ for $i \in V_\text{a}$ and $j \in V_\text{nca}$ to the outgoing communication channels of agents in the set $ V_\text{a}$ to agents that belong to $V_\text{nca}$.
\end{lem}

\begin{proof}
Let us define $\check{x}_i^\text{c}(t)= [x^\text{c}_i (t)^\top , \, \hat{x}^\text{c}_i (t)^\top]^\top$ for $i \in V_\text{a}$ and $x_\text{c} (t)= [x^\text{c}_{N_\text{nc}+1} (t)^\top , \, x^\text{c}_{N_\text{nc}+2} (t)^\top, \dots , \, x^\text{c}_{N} (t)^\top]^\top$. Consequently, the dynamics of $x^\text{c} (t)$ can be derived as follows:
\begin{equation}\label{e:lem_covert1}
\dot{x}_\text{c} (t) = A_\text{a}x_\text{c} (t)+A_\text{anc} x_\text{nc} (t).
\end{equation}

By adding the cyber attack signals $\hat{a}_{\hat{x}}^{ij} (t)$ and $\hat{a}_{y}^{ij} (t)$ to the outgoing communication channels of agents in the set $V_\text{a}$, \eqref{e:MAS_nc} can be reformulated in the following form:
\begin{equation}\label{e:lem_covert2}
\dot{x}_\text{nc} (t) = A_\text{nc} x_\text{nc} (t)+A_\text{nca}x_\text{c} (t) .
\end{equation}

One can augment \eqref{e:lem_covert1} and \eqref{e:lem_covert2} as follows:
\begin{equation}\label{e:lem_covert3}
\dot{\bar{x}}(t)=(I_N \otimes \check{A}+L \otimes (\check{B} \check{K}-\check{F} \check{C}))\bar{x}(t),
\end{equation}

\noindent where $\bar{x}(t)=[x_\text{nc} (t)^\top , \, x_\text{c} (t)^\top]^\top$. If the observer-based control protocol is designed such that the MAS reaches a consensus, then the augmented dynamics in \eqref{e:lem_covert3} also reaches a consensus. Hence, according to the Definition \ref{def:undetectable}, the cyber attacks on nodes in $V_\text{da}$ are undetectable by the entire MAS and their impacts on agents in $V_\text{nc}$ are eliminated.
\end{proof}

\section{Event-Triggered Cyber Attack Detection Methodology}\label{s:event}

\subsection{Event-Triggered Detector Module}
Our objective in this section is to interrupt disclosure capabilities of the adversary by employing an event-triggered protocol of information exchange among our proposed detector modules.

The event-triggered detector for the $i$-th agent is designed as follows:
\begin{equation}\label{e:detector_i}
\begin{split}
\dot{z}_i (t)=& A_\text{z} {z}_i (t) + B_\text{z} \hat{x}_i (t) +F_\text{z} \sum_{ j \in \mathcal{N}_i} (e^{A_\text{z}(t-t^j_{k_j})} z_j (t^j_{k_j}) \\
&-e^{A_\text{z}(t-t^i_{k_i})} z_i (t_{k_i}^i) +q_{i}^\text{z} a_\text{z}^{ji}(t^j_{k_j})), \, \, \, \, i=1, \dots, N,
\end{split}
\end{equation}

\noindent where ${z}_i (t) \in \mathbb{R}^n$ is the state of detector for agent $i$, $t_{k_i}^i$ denotes the time of the most recent triggering event of the agent $i$, $k_i \in \mathbb{N}$ indicates the $k_i$-th event on the agent $i$, $z_i (t_{k_i}^i)$ denotes the latest broadcast state of the detector for the $i$-th agent, $A_\text{z}$ is a diagonal Hurwitz matrix, and the matrices $(A_\text{z}, \, B_\text{z}, \, F_\text{z})$ are of appropriate dimensions that should be designed. Also $q_{i}^\text{z}=1$ indicates that the $i$-th detector is under cyber attacks, and $q_{i}^\text{z}=0$ if it is not under attack, and $ a_\text{z}^{ji}(t^j_{k_j}) \in \mathbb{R}^n$ denotes the cyber attack signal on the received information from the neighboring agents. It is worth noting that given the detector \eqref{e:detector_i} and the consensus-based observer \eqref{e:obs_i} the agents receive output measurement information, observer states, and detector states from their neighboring agents.

Similar to \cite{zhang2013observer}, we define the state error for the $i$-th agent as
\begin{equation}\label{e:z_error}
e^\text{z}_i(t)= e^{A_\text{z}(t-t^i_{k_i})} z_i(t_{k_i}^i)-z_i (t), \, \, \, \, t \in [t_{k_i}^i, \, t_{k_i+1}^i).
\end{equation}

\noindent Moreover, the triggering function on the agent $i$ is defined as
\begin{equation}\label{e:z_func}
f_i^\text{z} (t, e^\text{z}_i(t))=\|e^\text{z}_i(t) \|-c_\text{z} e^{-\alpha t},
\end{equation}

\noindent where $c_\text{z}$ and $\alpha$ are positive constants to be selected and designed. The triggering function \eqref{e:z_func} determines the occurrence of an event by the agent $i$. Hence, $f_i^\text{z} (t, e^\text{z}_i(t)) \geq 0$ implies that an event is triggered by the agent $i$. Consequently, this agent updates its detector's state such that $z_i(t_{k_i+1}^i)=z_i (t)$, and $t_{k_i+1}^i=t$. Subsequently, the updated state $z_i(t_{k_i+1}^i)$ is broadcast to agents that agent $i$ belongs to their neighboring set and they use the updated state in their detectors. Due to an event by the $i$-th agent, the state error \eqref{e:z_error} of this agent is reset to zero.

One can augment states of detectors into the vector $z(t)=[z_1(t)^\top , \, z_2(t)^\top , \dots, z_N (t)^\top]^\top$ such that dynamics of detectors can be expressed as follows:
\begin{equation}\label{e:z}
\begin{split}
\dot{z}(t)= & (I_N \otimes A_\text{z}) z(t) +(I_N \otimes B_\text{z}) \hat{x}(t)- L \otimes F_\text{z}(e_\text{z}(t) \\
&+z(t)) +Q_\text{z} \otimes a_\text{z}(t) ,
\end{split}
\end{equation}

\noindent where $\hat{x}(t)=[\hat{x}_1(t)^\top, \, \hat{x}_2(t)^\top , \dots, \, \hat{x}_N (t)^\top]^\top$, $e_\text{z}(t)=[e_1^\text{z}(t)^\top , \, e_2^\text{z}(t)^\top , \dots, e_N^\text{z} (t)^\top]^\top$, $Q_\text{z}=\text{diag}(q_1^\text{z}, \, q_2^\text{z}, \dots , \, q_N^\text{z})$, $a_\text{z}(t)=[a_\text{z}^1 (t)^\top, \, a_\text{z}^2 (t)^\top, \dots, a_\text{z}^N (t)^\top]^\top$, and $a_\text{z}^i (t)=\sum_{j \in \mathcal{N}_i} a_\text{z}^{ji}(t^j_{k_j})$ for $i=1, \dots, N$.

\begin{theorem}\label{th:triggerd}
	Let the Assumptions \ref{assu:spanning} and \ref{assu:observer_design} hold. Consider the MAS \eqref{e:sys_aug_1} and the detector \eqref{e:detector_i} under cyber attack free conditions, with the triggering function parameters \eqref{e:z_func} satisfy $0<c_\text{z}$ and $0<\alpha< -\text{max} \, \text{Re}(\lambda(\tilde{A}_\text{z}))$, where $\lambda(\tilde{A}_\text{z})$ denotes the eigenvalue of $\tilde{A}_\text{z}=I_{N} \otimes {A}_\text{z}-\Delta_\text{z} \otimes F_\text{z}$, and $\Delta_\text{z}$ can be computed in a similar manner as described for $\Delta$ in the proof of Theorem \ref{th:root_attack}, where the Laplacian matrix is now $L$. The detectors in \eqref{e:z} reach a consensus if and only if $A_\text{z}- \lambda_i F_\text{z}$, $i=2, \dots , N$, are Hurwitz, where $\lambda_i \neq 0$ is the $i$-th eigenvalue of $L$. Moreover, the detector \eqref{e:z} does not exhibit Zeno behavior under attack free conditions.
\end{theorem}

\begin{proof}
Under cyber attack free conditions the expression \eqref{e:z} can be rewritten as
\begin{equation*}
\begin{split}
\dot{z}(t)= & (I_N \otimes A_\text{z}-L \otimes F_\text{z}) z(t) +(I_N \otimes B_\text{z}) \hat{x}(t) \\
& - (L \otimes F_\text{z})e_\text{z}(t),
\end{split}
\end{equation*}

Following along the derivations in \cite{yang2016decentralized}, and since by definition the triggering function \eqref{e:z_func} does not cross zero in the interval $t \in [t_{k_i}^i, \, t_{k_i+1}^i)$, we have $\|e^\text{z}_i(t) \|< c_\text{z} e^{-\alpha t}$, which implies that $\|e_\text{z}(t)\| < \sqrt{N} c_\text{z} e^{-\alpha t}$. Therefore, it can be concluded that $\|e_\text{z}(t)\| \rightarrow 0$ as $t \rightarrow \infty$.

In a similar manner as in the proof of Theorem \ref{th:root_attack}, let us define the disagreement vector $\delta_\text{z}(t)=z(t)-(\mathbf{1}_{N} {r}^\top \otimes I_n)z(t)$. Since under cyber attack free conditions states of the observers $\hat{x}(t)$ reach a consensus and $\mathbf{1}_N$ is the right eigenvector corresponding to the zero for $I_N-\mathbf{1}_{N} {r}^\top$, it follows that $\dot{\delta}_{\text{z}}(t)$ is independent of $\hat{x}(t)$ once the agents reach a consensus.

Let us define $\varepsilon_\text{z} (t)=(T^{-1}_\text{z} \otimes I_{n})\delta_\text{z} (t)=[\varepsilon_1^\text{z} (t)^\top, \, \varepsilon_{2:N}^\text{z} (t)^\top]^\top$. Similar to the proof of the Theorem \ref{th:root_attack} it can be shown that $	\varepsilon_1^\text{z} (t) = \mathbf{0}_{n}$ and
\begin{equation}\label{e:e_z}
\dot{\varepsilon}_{2:N}^\text{z} (t) = (I_{N} \otimes {A}_\text{z}-\Delta_\text{z} \otimes {F}_\text{z})\varepsilon_{2:N}^\text{z} (t)-(\Delta_\text{z} W_\text{z} \otimes F_\text{z}) e_{2:N}^\text{z}(t),
\end{equation}

\noindent where $e_{2:N}^\text{z}(t)= [e_2^\text{z}(t)^\top , \dots, e_N^\text{z} (t)^\top]^\top$, $T_\text{z} \in \mathbb{R}^{ N \times N}$, $Y_\text{z} \in \mathbb{R}^{N \times N-1}$, $W_\text{z} \in \mathbb{R}^{N-1 \times N}$, and the block diagonal matrix $\Delta_\text{z} \in \mathbb{R}^{N-1 \times N-1}$ has diagonal entries that are equal to the nonzero eigenvalues of ${L}$ such that $T_\text{z}= [\mathbf{1}_{N} \, Y_\text{z}], \, T^{-1}_\text{z}= \begin{bmatrix}
{r}^\top \\
W_\text{z}
\end{bmatrix}$ and $
 T^{-1}_\text{z} {L} T_\text{z}= J_\text{z}=\begin{bmatrix}
0 & \mathbf{0}_{1 \times N-1} \\
\mathbf{0}_{N-1} & \Delta_\text{z}
\end{bmatrix}$.

Following along the proof of Theorem \ref{th:root_attack}, it can be shown that \eqref{e:e_z} is stable if and only if $A_\text{z}- \lambda_i F_\text{z}$, for $i=2, \dots , N$, are Hurwitz.

To show that \eqref{e:z} does not exhibit the Zeno behavior one needs to show that the inter-event intervals are lower bounded. The solution to \eqref{e:e_z} can be derived as follows:
\begin{equation*}
\varepsilon_{2:N}^\text{z} (t)= e^{\tilde{A}_\text{z} t} \varepsilon_{2:N}^\text{z} (0)+ \int_{0}^{t} e^{\tilde{A}_\text{z}(t- \tau)}(\Delta_\text{z} W_\text{z} \otimes F_\text{z}) e_{2:N}^\text{z}(\tau) \text{d}\tau ,
\end{equation*}

\noindent where $\tilde{A}_\text{z}=I_{N} \otimes {A}_\text{z}-\Delta_\text{z} \otimes F_\text{z}$. From the Lemma \ref{lem:event}, it can be concluded that
\begin{equation}\label{e:e_norm}
\|\varepsilon_\text{z} (t) \|=\|\varepsilon_{2:N}^\text{z} (t)\| \leq \alpha_3 e^{\lambda_{\tilde{A}_\text{z}}^\text{m} t} + \alpha_2 e^{-\alpha t},
\end{equation}

\noindent where $\alpha_3=\alpha_1+\alpha_2$, $\alpha_1= c_{\tilde{A}_\text{z}} nN \|P_{\tilde{A}_\text{z}} \|  \|P_{\tilde{A}_\text{z}}^{-1} \| \|\varepsilon_\text{z} (0)\| $, $\alpha_2 = (c_{\tilde{A}_\text{z}}nN \|P_{\tilde{A}_\text{z}} \|  \|P_{\tilde{A}_\text{z}}^{-1} \| \|\Delta_\text{z} W_\text{z} \otimes F_\text{z}\| \sqrt{N} c_\text{z} )/(|\alpha + \lambda_{\tilde{A}_\text{z}}^\text{m}|)$, $c_{\tilde{A}_\text{z}} >0$ is a positive constant, and $\text{max} \, \text{Re}(\lambda(\tilde{A}_\text{z})) < \lambda_{\tilde{A}_\text{z}}^\text{m}< -\alpha <0$. It follows from \eqref{e:e_norm} that $\|\delta_\text{z} (t) \| \leq \|T_\text{z} \otimes I_n\| \|\varepsilon_\text{z} (t)\| \leq \beta_1 e^{\lambda_{\tilde{A}_\text{z}}^\text{m} t}+ \beta_2 e^{-\alpha t}$, where $\beta_1= \|T_\text{z} \otimes I_n\| \alpha_3$ and $\beta_2= \|T_\text{z} \otimes I_n\| \alpha_2$.

In the interval $t \in [t_{k_i}^i, \, t_{k_i+1}^i)$ the dynamics of $e_\text{z} (t)$ can be expressed as
\begin{equation*}
\dot{e}_\text{z} (t)= (I_N \otimes A_\text{z}- L \otimes F_\text{z}) e_\text{z} (t)-(I_N \otimes B_\text{z}) \hat{x}(t)- (L \otimes F_\text{z}) z(t).
\end{equation*}

\noindent Let us define $\delta_\text{e}(t)=e_\text{z} (t)-(\mathbf{1}_{N} {r}^\top \otimes I_n)e_\text{z} (t)$, which is governed by
\begin{equation*}\label{e:e_delta}
\begin{split}
\dot{\delta}_\text{e} (t)= & (I_N \otimes A_\text{z}- L \otimes F_\text{z}) \delta_\text{e}(t)-((I_N-\mathbf{1}_{N} {r}^\top) \otimes B_\text{z}) \hat{x}(t) \\
& - (L \otimes F_\text{z}) \delta_\text{z}(t).
\end{split}
\end{equation*}

\noindent Since it is assumed that the MAS is cyber attack free, $\delta_{\text{z}}(t)$ does not approach to zero unless $((I_N-\mathbf{1}_{N} {r}^\top) \otimes B_\text{z}) \hat{x}(t)$ approaches to zero. Thus, there exists a bounded scalar $M >0$ such that $\|((I_N-\mathbf{1}_{N} {r}^\top) \otimes B_\text{z}) \hat{x}(t) \| \leq M \|\delta_{\text{z}}(t)\|$. Moreover, given that $\|I_N-\mathbf{1}_{N} {r}^\top\|\geq \| I_N\|$, by definition one can conclude that $\|e_\text{z} (t) \| \leq \| \delta_\text{e} (t) \|$. Therefore, it implies  that
\begin{equation}
\begin{split}
\|\dot{e}_\text{z} (t) \| \leq & \| \dot{\delta}_\text{e} (t) \| \leq \|I_N \otimes A_\text{z} - L \otimes F_\text{z}\| \|I_N-\mathbf{1}_{N} {r}^\top\| \\
& \times \|e_\text{z}(t)\| +\|L \otimes F_\text{z}\| \|\delta_{\text{z}} (t) \|\\
& + \|((I_N-\mathbf{1}_{N} {r}^\top) \otimes B_\text{z}) \hat{x}(t) \| \leq  g_\text{z}(t), \nonumber
\end{split}
\end{equation}

\noindent where $g_\text{z}(t) \triangleq  a_1 e^{\lambda_{\tilde{A}_\text{z}}^\text{m} t} + a_2 e^{-\alpha t} $, $a_1=(\|L \otimes F_\text{z}\|+M) \beta_1$, and $a_2=(\|L \otimes F_\text{z}\|+M) \beta_2 + \sqrt{N} \|I_N \otimes A_\text{z} - L \otimes F_\text{z}\| \|I_N-\mathbf{1}_{N} {r}^\top\|$.

Let $t^*$ denote the latest triggering instant and consider $\tau^*=t-t^*$ as the time-interval between the two latest triggered events. Given that at the triggering instant $f_i^\text{z} (t, e^\text{z}_i(t))=0$, it can be concluded that in the $i$-th detector the next event cannot be triggered before $\| e^\text{z}_i (t) \| =c_\text{z} e^{-\alpha t}$, that implies $\| e_\text{z} (t) \| =\|\int_{t^*}^{t}\dot{e}_\text{z} (s) \text{d}s \| \leq \int_{t^*}^{t}g_\text{z}(s) (s) \text{d}s =  \sqrt{N} c_\text{z} e^{-\alpha t}$. Since $t^* \geq t$, we have $e^{ -\alpha t} \leq e^{-\alpha t^*} $ and $e^{\lambda_{\tilde{A}_\text{z}}^\text{m} t} \leq e^{\lambda_{\tilde{A}_\text{z}}^\text{m} t^*} $. Consequently, $\tau^*$ is lower bounded by $\bar{\tau}$, which is the solution to the equation $(a_1 e^{\lambda_{\tilde{A}_\text{z}}^\text{m} t^*} + a_2 e^{-\alpha t^*} )\bar{\tau}=\sqrt{N} c_\text{z} e^{-\alpha (t^*+\bar{\tau})}$, that is equivalent to $(a_1 e^{(\lambda_{\tilde{A}_\text{z}}^\text{m}+\alpha) t^*} + a_2 )\bar{\tau}=\sqrt{N} c_\text{z} e^{-\alpha \bar{\tau}}$.

Since $0<\alpha< |\lambda_{\tilde{A}_\text{z}}^\text{m}|< -\text{max} \, \text{Re}(\lambda(\tilde{A}_\text{z}))$, there exists $\tilde{\tau}$ such that $\tau^* \geq \bar{\tau} \geq \tilde{\tau}$, where $\tilde{\tau}$ is the strictly positive solution to the equation $(a_1 + a_2 )\tilde{\tau}=\sqrt{N} c_\text{z} e^{-\alpha \tilde{\tau}}$. Hence, $\tilde{\tau}$ is the lower bound on the inter-event times of the detector \eqref{e:detector_i}, which implies that there are no Zeno behavior. This completes the proof of the theorem.
\end{proof}

\begin{definition}\label{def:res}
	A cyber attack injected to the closed-loop MAS \eqref{e:agent_i_attack} and the observer \eqref{e:obs_i_attack} is detected if the residual signal
	\begin{equation}\label{e:res_z}
	res_\text{z}^i(t)= \sum_{j \in \mathcal{N}_i} \| z_j(t_{k_j}^j) - z_i (t_{k_i}^i) +q_{ji} a_\text{z}^{ji} (t) \|,
	\end{equation}
	\noindent satisfies the inequality $\|res_\text{z}^i(t)\| > \eta_\text{z}$, where $\eta_\text{z}$ is the cyber attack detection threshold.
\end{definition}

\begin{assumption}\label{assu:NotKnow}
	The adversaries do not have knowledge on the parameters of the event-triggered detector \eqref{e:detector_i}. Hence, they are not capable of designing $a_\text{z}^{ji}(t^j_{k_j})$ such that $z_j (t^j_{k_j})-z_i (t_{k_i}^i) +q_{ji} a_\text{z}^{ji}(t^j_{k_j})=0$ as $t^j_{k_j} \rightarrow \infty$.
\end{assumption}

\begin{remark}
	Since agents communicate their state detectors according to the triggering function \eqref{e:z_func}, the adversary does not have access to these states continuously. Consequently, it is quite reasonable to assume that the adversary does not have knowledge of the exact values of the parameters in \eqref{e:detector_i}.
\end{remark}

\begin{corollary}\label{cor:event_trig}
	Consider the Assumptions \ref{assu:spanning}-\ref{assu:NotKnow} hold and the MAS \eqref{e:sys_aug_1} is under quasi-covert cyber attacks as introduced in the Lemma \ref{lem:quasi_covert}. Given that the triggering function \eqref{e:z_func} has parameters that are provided in the Theorem \ref{th:triggerd}, let the cyber attack signal be denoted by $a_\text{z}^{ji}(t^j_{k_j})=a_\text{z0}-e^{A_\text{z}(t-t^j_{k_j})} z_j (t^j_{k_j})$ for $i \in V_\text{da}$ and $j \in V_\text{nc}$, where $a_\text{z0} \in \mathbb{R}^n$ is a constant vector. Consequently, the generated residual \eqref{e:res_z} by the $k$-th agent is nonzero if $k \in V_\text{a}$ or $k \in V_\text{nca}$, and $B_\text{z}$ is full column rank, and $A_\text{z}- \lambda_q F_\text{z}$, for $q=2, \dots , N$, are Hurwitz in detectors \eqref{e:z}, where $\lambda_q$ is defined as in Theorem \ref{th:triggerd}.
\end{corollary}

\begin{proof}
    Suppose $B_\text{z}$ is a full column rank matrix and $A_\text{z}- \lambda_q F_\text{z}$, for $q=2, \dots , N$, are Hurwitz. Consequently, following along the steps as in proof of Theorem \ref{th:detectable}, the detectors that belong to the set $V_\text{a}$ do not reach $a_{z0}$ as the consensus set point since dynamics of the virtual adversary agent does not have $B_\text{z}\hat{x}_i$. Hence, cyber attacks are detectable by using the detectors of agents that belong to this set.

    Since under cyber attacks the agents observers states  do not reach a consensus, $\dot{\delta}_{\text{z}}(t)$ in the proof of Theorem \ref{th:triggerd} depends on $\hat{x}(t)$, which implies that the detectors do not reach a consensus as well and the cyber attacks are detectable on the set $V_\text{nca}$. This completes the proof of the corollary.
\end{proof}

\section{Numerical Simulation Case Study}\label{s:example}
In this case study, cyber attacks that are introduced in the Theorem \ref{th:detectable} and the quasi-covert cyber attacks in the Lemma \ref{lem:quasi_covert} are investigated. Moreover, the effectiveness of event triggered detector that was proposed in Section \ref{s:event} in detecting the quasi-covert cyber attacks is demonstrated and illustrated. A MAS consisting of 6 agents along with their observer-based consensus protocols having the dynamics as given in \eqref{e:agent_i} and \eqref{e:obs_i}, respectively, with the following parameters are studied \cite{li2009consensus} in:
\begin{equation*}
\begin{split}
A & =\begin{bmatrix}
-2 & 2 \\
-1 & 1
\end{bmatrix}, \, B=\begin{bmatrix}
1 \\
0
\end{bmatrix}, \, C=\begin{bmatrix}
1 & 0 \\
0 & 1
\end{bmatrix}, \\
F &= \begin{bmatrix}
15 & 0 \\
15 & 15
\end{bmatrix}, \, K=\begin{bmatrix}
2 & -10
\end{bmatrix}, \, c=-2.
\end{split}
\end{equation*}

\noindent The Laplacian matrix of the communication links among agents is given by $L=$[1, -1, 0, 0, 0, 0; -1, 1, 0, 0, 0, 0; -1, 0, 2, -1, 0, 0; 0, -1, 0, 2, 0, -1; 0, 0, -1, 0, 1, 0; 0, 0, 0, 0, -1, 1]. The detector parameters in \eqref{e:detector_i} are $A_\text{z}  =\begin{bmatrix}
-1 & 0 \\
0 & -2
\end{bmatrix}$, $B_\text{z}=I_2 \times 5$, $F_\text{z}= I_2 \times 10$, $c_\text{z}=0.002$, and $\alpha=0.2$.

Agents 1 and 2 are the roots of the spanning trees as contained in graph $\mathcal{G}$ of the network. In this case study, agents 4 and 5 are directly under cyber attacks such that $V_\text{da}=\{4,5\}$, $V_\text{a}=\{4,5,6\}$, $V_\text{nc}=\{1,2,3\}$, and $V_\text{nca}=\{3\}$.

In Fig. \ref{fig:res}, the residuals of agents in presence of cyber attacks as introduced in the Theorem \ref{th:detectable} are shown. After the occurrence of the cyber attack at $t=10$ (s) it is detected by the agent 3 that belongs to the set $V_\text{nca}$, whereas it is undetectable by the rest of network agents.

The impacts of the quasi-covert cyber attack on MAS are illustrated in Fig. \ref{fig:x_cov}. As can be observed in this figure the agent 3 states after the quasi-covert cyber attack is injected at $t=30$ (s) reach to those of the agents 1 and 2. Consequently, this cyber attack is undetectable on the entire network as shown in Fig. \ref{fig:res_cov}. By using the event-triggered detectors that are proposed in \eqref{e:detector_i}, the residuals \eqref{e:res_z} are generated and shown in Fig. \ref{fig:res_z_cov}. It follows from this figure that all agents' residuals that belong to the set $V_\text{a}$ are nonzero and exceed the threshold $\eta_\text{z}=3$. Moreover, the agent 3 residual exceeds the threshold as well.

\begin{figure}[!t]
	\centering
\includegraphics[scale= 0.34]{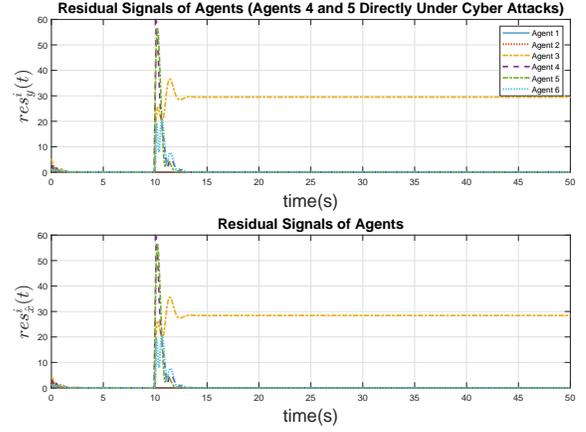}
	\caption{Residuals of agents while the cyber attacks are injected at $t=10$ (s).}
	\label{fig:res}
\end{figure}

\begin{figure}[!t]
	\centering
\includegraphics[scale= 0.34]{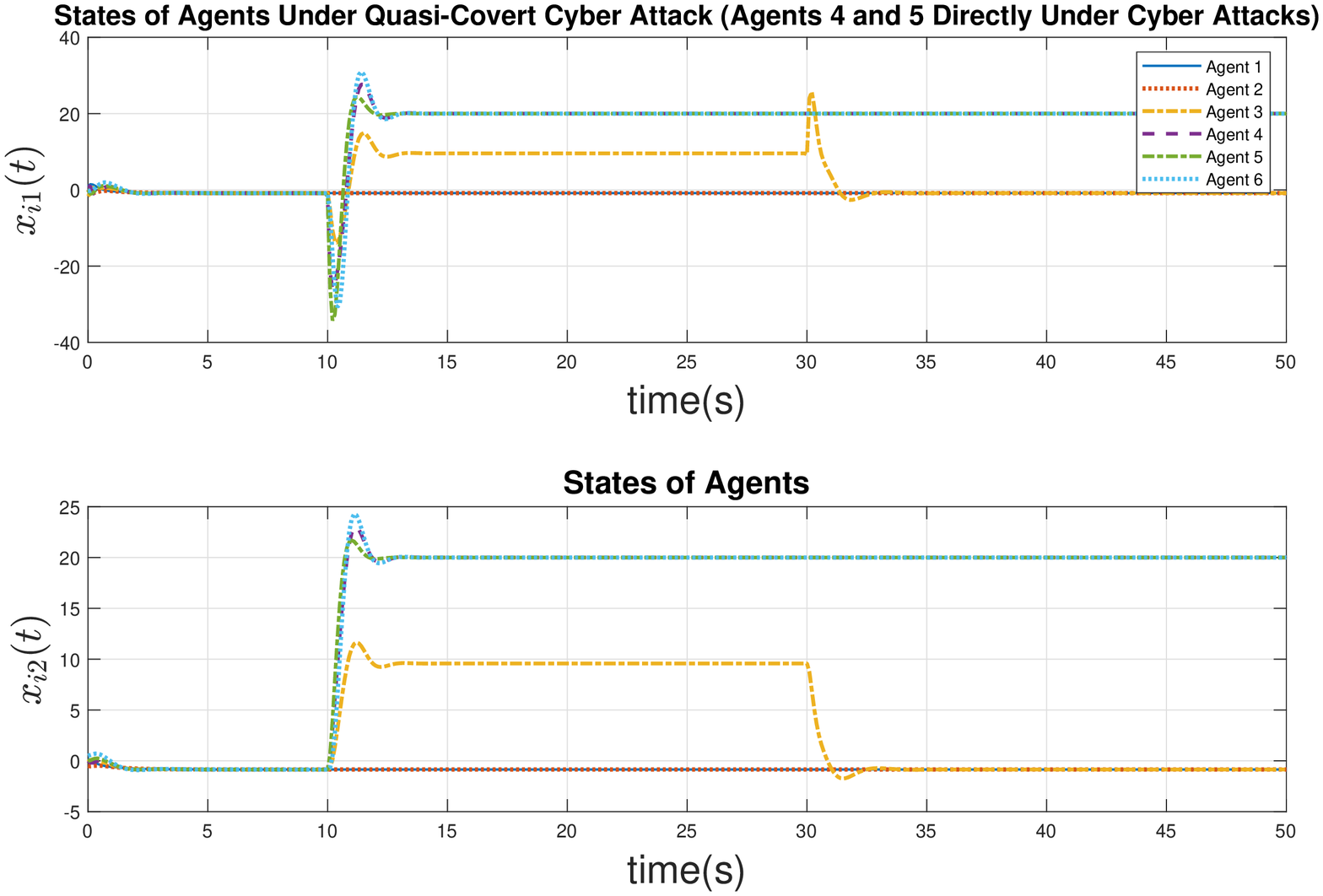}
	\caption{States of agents in presence of cyber attacks at $t=10$ (s) and quasi-covert cyber attacks at $t=30$ (s).}
	\label{fig:x_cov}
\end{figure}

\begin{figure}[!t]
	\centering
\includegraphics[scale= 0.34]{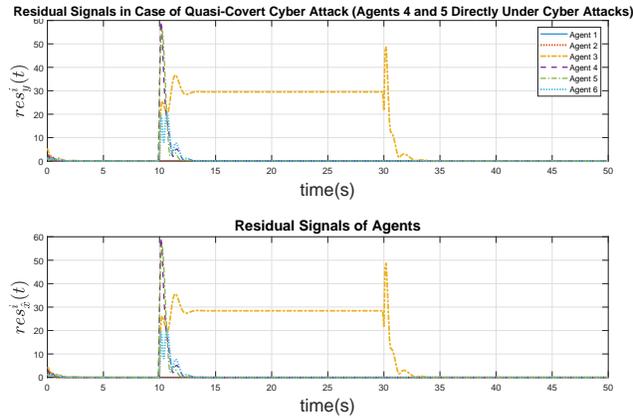}
	\caption{Residuals approach to zero after occurrence of the quasi-covert cyber attacks at $t=30$ (s).}
	\label{fig:res_cov}
\end{figure}

\begin{figure}[!t]
	\centering
\includegraphics[scale= 0.41]{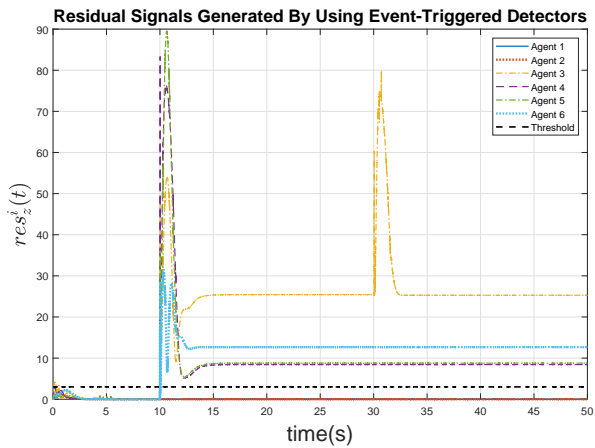}
	\caption{Residuals that are generated by using the event-triggered detectors and exceed the threshold in presence of quasi-covert cyber attacks injected at $t=30$ (s).}
	\label{fig:res_z_cov}
\end{figure}

\section{Conclusion}\label{s:conclu}
In this paper, detectability of cyber attacks on the communication links of certain teams of multi-agent systems has been investigated. A definition that can be used to specify undetectable cyber attacks on MAS was developed and proposed. It was shown that cyber attacks on communication links of the root of a directed spanning tree graph are undetectable. Moreover, cyber attacks on the non-root agents were investigated. It was shown that cyber attacks on non-root agents can be detected by the set of agents provided that one can determine an uncompromised directed path from the root of the graph to these agents. Novel quasi-covert cyber attacks were introduced that can be injected to maintain and ensure  these attacks on non-root agents remain undetectable by the entire network. Finally, an event-triggered detector was proposed that is capable of detecting quasi-covert cyber attacks. A challenging topic for future investigation is to study event-triggered detectors and detection of cyber attacks in presence of disturbances, noise, and uncertainties.

\bibliographystyle{IEEEtran}
\bibliography{IEEEabrv,CDCRef}

\end{document}